\documentclass[12pt]{article}

\usepackage[margin=1in]{geometry}
\usepackage{amsmath,amssymb,amsthm,mathtools,bm}
\usepackage{booktabs}
\usepackage{enumitem}
\usepackage{hyperref}

\usepackage{amsmath,amssymb,amsthm,mathtools,booktabs,bm,graphicx,microtype,array}
\usepackage[T1]{fontenc}
\usepackage{lmodern}
\usepackage{caption}

\usepackage{tikz}
\usepackage{pgfplots}
\pgfplotsset{compat=1.18}

\hypersetup{colorlinks=true,citecolor=blue,linkcolor=blue,urlcolor=blue}

\newtheorem{proposition}{Proposition}
\newtheorem{theorem}{Theorem}
\newtheorem{corollary}{Corollary}
\newtheorem{lemma}{Lemma}
\theoremstyle{definition}

\theoremstyle{remark}
\newtheorem{remark}{Remark}

\usepackage{lineno}

\newcommand{\E}{\mathbb{E}}
\newcommand{\Prob}{\mathbb{P}}
\newcommand{\Var}{\mathrm{Var}}
\newcommand{\Cov}{\mathrm{Cov}}
\newcommand{\Corr}{\mathrm{Corr}}
\newcommand{\ind}{\mathrm{ind}}
\newcommand{\dep}{\mathrm{dep}}
\newcommand{\st}{\mathrm{st}}

\makeatletter
\renewcommand{\maketitle}{%
	\begin{center}
		{\Large \bfseries \@title \par}
		\vskip 1em
		{\normalsize \@author \par}
		\vskip 0.5em
		{\@date \par}
	\end{center}
	\vskip 1em
}
\makeatother

\title{Exact analysis of a split--merge queue with latent Erlang-factor dependent subtask times}
\author{Dawit Zerom \\ California State University - Fullerton}
\date{}

\begin{document}
	
		%\linenumbers
		
	\maketitle

		\begin{abstract}
			 This paper studies a two-server split--merge queue with positively dependent
			subtask service times modeled through a latent-factor bivariate Erlang construction. An exact characterization of the split--merge completion time is obtained, including explicit formulas for its first two moments and the resulting mean waiting time. Under fixed marginal service-time distributions, independence is shown to
			stochastically increase the completion time and hence overestimate mean waiting
			time. Numerical illustrations show that this benchmark gap can be substantial.
				
			\smallskip\noindent
			\textbf{Keywords:} Split--merge queue; Dependent service times; Bivariate Erlang model; Stochastic order; Mean waiting time.
		\end{abstract}

%%%%%%%%%%%%%%%%%%%%%%%%%%%%%%

%% SECTION 1

%%%%%%%%%%%%%%%%%%%%%%%%%%%
	
	\section{Introduction}
	\label{sec:intro}
	
	In a two-server split--merge system, each arriving job is divided into two subtasks,
	one processed at each server, and the job departs only when both subtasks are complete.
	Because the next job cannot begin service until the current job has fully merged, the
	system is equivalent to an $M/G/1$ queue with service time
	\[
	M=\max\{X,Y\},
	\]
	where $X$ and $Y$ denote the two subtask service times. Thus, once the distribution
	of $M$ is characterized, standard $M/G/1$ formulas yield steady-state performance
	measures such as the mean waiting time and mean sojourn time.
	
	When $X$ and $Y$ are independent,
	\[
	\Prob(M\le t)=\Prob(X\le t,\;Y\le t)=F_X(t)F_Y(t),
	\]
	the distribution of $M$ is determined directly by the marginals. This yields explicit
	results in several settings, including phase-type and Erlang models; see, for example,
	Fiorini and Lipsky~\cite{FioriniLipsky2015}. 
	Under dependence, however, the event
	$\{X\le t,\;Y\le t\}$ no longer factorizes, so characterizing the distribution of $M$
	requires the joint law of $(X,Y)$, and exact closed-form results are correspondingly scarce. To our knowledge, existing exact treatments are essentially confined to simple Markovian special cases \cite{KoSerfozo2004}.
	 More generally, analyses of split--merge and related synchronization systems beyond the simplest independent setting have often proceeded through bounds, approximations, or asymptotic methods ~\cite{BaccelliMakowski1989,BaccelliMakowskiShwartz1989,VarmaMakowski1994,
		TanKnessl1996,QiuPerezHarrison2015,Tsimashenka2016,Thomasian2014,NelsonTantawi1988}.
			
Yet dependence between $X$ and $Y$ is natural and arguably the norm in many
split--merge applications. In storage systems, manufacturing settings, and
parallel data services, the two subtasks generated by the same job are often
influenced by a common latent characteristic such as job size, complexity, or
priority. Such shared job-specific factors induce positive same-job dependence
across servers even when the server-specific processing components differ. The
challenge, therefore, is to incorporate this dependence structure in a way that
remains analytically tractable.
    
	Motivated by this, we study a dependent split--merge model based on a
	parsimonious bivariate Erlang construction. The subtask times $(X,Y)$ are generated from
	shared latent Erlang factors through a positive linear map, so dependence arises from
	common latent workload components while preserving analytically useful structure. This
	construction is therefore appealing on both modeling and technical grounds: it captures a
	natural source of positive dependence and still permits an exact reduction of the joint
	event $\{X\le t,\;Y\le t\}$ to a one-dimensional calculation. As a result, the
	split--merge completion time $M=\max\{X,Y\}$ remains analytically tractable.
	
	For this model, we derive explicit expressions for the distribution of $M$, its first two
	moments, and the corresponding mean waiting time. The exact formulas also enable a clean
	comparison with the benchmark model that imposes independence while preserving the same
	marginal distributions. We show that the proposed dependence structure implies positive
	quadrant dependence, which in turn yields
	\[
	M \le_{\mathrm{st}} M^{\mathrm{ind}},
	\]
	where $M^{\mathrm{ind}}$ denotes the completion time under the corresponding independence
	benchmark. Consequently, the independence benchmark is conservative for mean waiting time.
	
	The contribution of the paper is therefore twofold. First, it provides an exact analysis
	of a dependent split--merge model whose dependence mechanism is both interpretable and
	tractable. Second, it establishes a structural same-marginal comparison showing that
	independence systematically overestimates completion times and mean waiting times under
	this form of dependence. Numerical illustrations show that the resulting benchmark gap
	can be substantial across different loads and parameter specifications.
	
	The remainder of the paper is organized as follows.
	Section~\ref{sec:model} introduces the dependent bivariate Erlang model.
	Section~\ref{sec:exact} derives the distribution and moments of $M$.
	Section~\ref{sec:comparison} establishes the comparison with the independence benchmark.
	Section~\ref{sec:numerics} presents numerical results.
%	, and Section~\ref{sec:conclusion}
%	concludes.

	%%%%%%%%%%%%%%%%%%%%%%%%%%%
	%%% SECTION 2 
	%%%%%%%%%%%%%%%%%%%%%%%%%

\section{A tractable dependent model for subtask service times}
\label{sec:model}

We model the subtask service-time pair $(X,Y)$ by a bivariate Erlang construction that
induces positive same-job dependence while preserving analytical tractability for the
split--merge completion time $M=\max\{X,Y\}$.

\medskip
\noindent
{\bf Construction.}
Let
\[
U\sim \mathrm{Erlang}(\alpha_1,\lambda_1), \qquad
V\sim \mathrm{Erlang}(\alpha_2,\lambda_2),
\]
be independent, where $\alpha_1,\alpha_2\in\mathbb{N}$ and $\lambda_1,\lambda_2>0$.
We define
\begin{equation}
	X=aU+bV,\qquad Y=cU+dV,
	\label{eq:model}
\end{equation}
with $a,b,c,d>0$. Thus $X$ and $Y$ are increasing functions of the same latent Erlang
factors, which induces positive dependence through shared job-level workload components. 

\medskip
\noindent
{\bf Support and joint density.}
To ensure that the linear map \eqref{eq:model} is invertible, we assume
\begin{equation*}
	\Delta:=ad-bc>0.
	%\label{eq:detcond}
\end{equation*}
Then
\begin{equation}
	u=\frac{dx-by}{\Delta},\qquad
	v=\frac{ay-cx}{\Delta}.
	\label{eq:inversemap}
\end{equation}
Since $U,V\ge 0$, a point $(x,y)$ belongs to the support of $(X,Y)$ if and only if
\[
dx-by\ge 0,\qquad ay-cx\ge 0,
\]
or equivalently,
\begin{equation*}
	\frac{c}{a}\,x\le y\le \frac{d}{b}\,x.
%	\label{eq:wedgeineq}
\end{equation*}
Hence the support is the wedge
\begin{equation*}
	\mathcal{W}=\left\{(x,y):x>0,\; \frac{c}{a}\,x\le y\le \frac{d}{b}\,x\right\}.
%	\label{eq:wedge}
\end{equation*}

Because \eqref{eq:inversemap} is a linear bijection from $\mathcal{W}$ onto
$\mathbb{R}_{\ge 0}^2$ with Jacobian $1/\Delta$, the change-of-variables formula yields
\begin{align*}
	f_{X,Y}(x,y)
	&=
	\frac{\lambda_1^{\alpha_1}\lambda_2^{\alpha_2}}
	{(\alpha_1-1)!(\alpha_2-1)!\,\Delta^{\alpha_1+\alpha_2-1}}
	(dx-by)^{\alpha_1-1}(ay-cx)^{\alpha_2-1}
	\notag\\
	&\quad \times
	\exp\!\left(
	-\frac{\lambda_1(dx-by)+\lambda_2(ay-cx)}{\Delta}
	\right),
	\qquad (x,y)\in\mathcal{W},
%	\label{eq:density}
\end{align*}
and $f_{X,Y}(x,y)=0$ for $(x,y)\notin\mathcal{W}$. 

\medskip
\noindent
{\bf Dependence properties.}
By independence of $U$ and $V$,
\begin{equation*}
	\Cov(X,Y)
	=
	ac\,\Var(U)+bd\,\Var(V)
	=
	\frac{ac\,\alpha_1}{\lambda_1^2}
	+
	\frac{bd\,\alpha_2}{\lambda_2^2}
	>0.
%	\label{eq:covariance}
\end{equation*}
Moreover, $(X,Y)$ is a coordinatewise nondecreasing transformation of the independent
pair $(U,V)$. Hence $(X,Y)$ is associated in the sense of \cite{Esary1967}, and in
particular satisfies positive quadrant dependence:
\begin{equation}
	\Prob(X\le x,\;Y\le y)\ge F_X(x)F_Y(y),
	\qquad x,y\ge 0.
	\label{eq:pqd}
\end{equation}
This property will be used in Section~\ref{sec:comparison}. 

\medskip
\noindent
{\bf Relation to the literature.}
The construction \eqref{eq:model} belongs to the class of shared-factor
multivariate gamma models, in which dependent nonnegative random variables
are built from linear combinations of independent gamma components
~\cite{MathaiMoschopoulos1992}. Related Erlang-mixture families are discussed
in~\cite{LeeLin2012,WillmotWoo2015}.
In our setting, its linear structure yields an explicit inverse map and simple support geometry, which make the exact analysis in Section~\ref{sec:exact} possible.

		%%%%%%%%%%%%%%%%%%%%%%%%%%%
	%%% SECTION 3
	%%%%%%%%%%%%%%%%%%%%%%%%%
	\section{Exact split--merge analysis}
	\label{sec:exact}
	
	Jobs arrive according to a Poisson process with rate $\nu$. Since each job departs only
	after both subtasks are complete, the effective service time is
	\[
	M=\max\{X,Y\},
	\]
	and the system reduces to an $M/G/1$ queue with service-time distribution equal to that
	of $M$. Hence
	\begin{equation*}
		G_M(t):=\Prob(M\le t)=\Prob(X\le t,\;Y\le t), \qquad t\ge 0.
	%	\label{eq:GMdiag}
	\end{equation*}
	Under the model of Section~\ref{sec:model},
	\begin{equation}
		G_M(t)=\Prob(aU+bV\le t,\;cU+dV\le t).
		\label{eq:latent_event}
	\end{equation}
	
	To exclude the degenerate cases in which one subtask dominates almost surely, we impose
	\begin{equation}
		a>c,\qquad d>b.
		\label{eq:crossing}
	\end{equation}
	Indeed, if $a\ge c$ and $b\ge d$, then $X\ge Y$ a.s. and $M=X$; similarly, if
	$c\ge a$ and $d\ge b$, then $M=Y$ a.s. Under \eqref{eq:crossing}, neither subtask
	dominates pointwise.
	
	Conditioning on $U=u$ in \eqref{eq:latent_event} gives
	\[
	V\le \frac{t-au}{b}, \qquad V\le \frac{t-cu}{d}.
	\]
	Since $V\ge 0$ and $a>c$, the feasible range is
	\begin{equation*}
		0\le u\le \frac{t}{a}.
	%	\label{eq:ufeasible}
	\end{equation*}
	On this interval, the active upper bound on $V$ is
	\[
	\min\!\left\{\frac{t-au}{b},\,\frac{t-cu}{d}\right\}.
	\]
	The two bounds coincide at
	\begin{equation*}
		u^*(t)=\frac{(d-b)t}{\Delta}, \qquad \Delta:=ad-bc>0.
	%	\label{eq:ustar}
	\end{equation*}
	When \eqref{eq:crossing} holds,
 one has $0<u^*(t)<t/a$, and
	\[
	\min\!\left\{\frac{t-au}{b},\,\frac{t-cu}{d}\right\}
	=
	\begin{cases}
		\dfrac{t-cu}{d}, & 0\le u\le u^*(t), \\[1.2ex]
		\dfrac{t-au}{b}, & u^*(t)\le u\le t/a.
	\end{cases}
	\]

	The exact distribution of $M$ follows by conditioning on $U$ with the integral split at the crossing point $u^*(t)$.
	
	\begin{theorem}
		\label{thm:main}
		Suppose \eqref{eq:crossing} holds. Then, for every $t\ge 0$,
		\begin{equation}
			G_M(t)
			=
			\int_0^{u^*(t)} F_V\!\left(\frac{t-cu}{d}\right)f_U(u)\,du
			+
			\int_{u^*(t)}^{t/a} F_V\!\left(\frac{t-au}{b}\right)f_U(u)\,du.
			\label{eq:GMfinal}
		\end{equation}
	\end{theorem}
	
	Since $U\sim \mathrm{Erlang}(\alpha_1,\lambda_1)$ and
	$V\sim \mathrm{Erlang}(\alpha_2,\lambda_2)$,
	\[
	f_U(u)=\frac{\lambda_1^{\alpha_1}}{(\alpha_1-1)!}\,u^{\alpha_1-1}e^{-\lambda_1u},
	\qquad
	F_V(v)=1-e^{-\lambda_2v}\sum_{k=0}^{\alpha_2-1}\frac{(\lambda_2v)^k}{k!}.
	\]

	Thus \eqref{eq:GMfinal} can be evaluated in closed form in terms of finite sums of incomplete gamma terms.

	\begin{corollary}
		\label{cor:perf}
		Under the conditions of Theorem~\ref{thm:main},
		\begin{equation}
			\E[M]=\int_0^\infty \bigl(1-G_M(t)\bigr)\,dt,
			\qquad
			\E[M^2]=2\int_0^\infty t\,\bigl(1-G_M(t)\bigr)\,dt.
			\label{eq:momentformulas}
		\end{equation}
		If moreover $\rho:=\nu\,\E[M]<1$, then
		\begin{equation}
			\E[W]=\frac{\nu\,\E[M^2]}{2(1-\rho)},
			\qquad
			\E[T]=\E[W]+\E[M].
			\label{eq:PK}
		\end{equation}
	\end{corollary}

	\begin{proof}
		Equation \eqref{eq:momentformulas} follows from standard tail-integral identities, and \eqref{eq:PK} follows from the Pollaczek--Khinchine formula for the corresponding $M/G/1$ queue.
	\end{proof}
	
	The next lemma evaluates the first two moments explicitly.
	
	\begin{lemma}
		\label{lem:moments}
		Under the conditions of Theorem~\ref{thm:main}, define
		\begin{equation}
			\mu:=\frac{a-c}{d-b}, \qquad \kappa:=\lambda_1+\lambda_2\mu.
			\label{eq:mukappa}
		\end{equation}
		Then
		\begin{equation}
			\E[M]
			=
			\frac{a\alpha_1}{\lambda_1}
			+
			\frac{b\alpha_2}{\lambda_2}
			+
			\frac{(d-b)\lambda_1^{\alpha_1}}{(\alpha_1-1)!}
			\sum_{k=0}^{\alpha_2-1}
			\frac{(\alpha_2-k)\mu^k\lambda_2^{\,k-1}}{k!}
			\cdot
			\frac{(\alpha_1+k-1)!}{\kappa^{\alpha_1+k}},
			\label{eq:EMclosed}
		\end{equation}
		and 
		\begin{align}
			\E[M^2]
			&=
			\frac{a^2\alpha_1(\alpha_1+1)}{\lambda_1^2}
			+
			\frac{b^2\alpha_2(\alpha_2+1)}{\lambda_2^2}
			+
			\frac{2ab\,\alpha_1\alpha_2}{\lambda_1\lambda_2}
			\notag\\
			&\quad
			+
			\frac{2(d^2-b^2)\lambda_1^{\alpha_1}}{(\alpha_1-1)!}
			\sum_{k=0}^{\alpha_2}
			\frac{S_k\mu^k\lambda_2^{\,k-2}}{k!}
			\cdot
			\frac{(\alpha_1+k-1)!}{\kappa^{\alpha_1+k}}
			\notag\\
			&\quad
			+
			\frac{2(dc-ab)\lambda_1^{\alpha_1}}{(\alpha_1-1)!}
			\sum_{k=0}^{\alpha_2-1}
			\frac{(\alpha_2-k)\mu^k\lambda_2^{\,k-1}}{k!}
			\cdot
			\frac{(\alpha_1+k)!}{\kappa^{\alpha_1+k+1}},
			\label{eq:EM2closed}
		\end{align}

		where
		\begin{equation}
			S_k:=\sum_{j=\max(0,k-1)}^{\alpha_2-1}(j+1).
			\label{eq:Sk}
		\end{equation}
	\end{lemma}
	
	\begin{proof}
		The result follows by substituting \eqref{eq:GMfinal} into \eqref{eq:momentformulas}, expanding the Erlang cdf of $V$, and evaluating the resulting finite sums of gamma-type integrals. The full derivation is provided in the Appendix.
	\end{proof}

	%%%%%%%%%%%%%%%%%%%%%%%%%%%%%%%%%%%
	
	%% SECTION 4 
	
	%%%%%%%%%%%%%%%%%%%%%%%%%%%%%%%%%

	\section{Comparison with an independence benchmark}
	\label{sec:comparison}
	
	We compare the dependent split--merge model with the benchmark obtained by preserving
	the marginals of $X$ and $Y$ while removing their dependence.
	
	\medskip
	\noindent
	{\bf Independence benchmark.}
	Let $(X^{\ind},Y^{\ind})$ satisfy
	\[
	X^{\ind}\stackrel{d}{=}X,\qquad Y^{\ind}\stackrel{d}{=}Y,
	\]
	with $X^{\ind}$ and $Y^{\ind}$ independent, and define
	\[
	M^{\ind}=\max\{X^{\ind},Y^{\ind}\}.
	\]
	Let $W^{\dep}$ and $W^{\ind}$ denote the steady-state waiting times in the corresponding
	$M/G/1$ queues with service times $M$ and $M^{\ind}$, respectively.
	
	\begin{proposition}
		\label{prop:benchmark}
		Under the model of Section~\ref{sec:model},
		\begin{equation}
			G_M(t)=\Prob(X\le t,\;Y\le t)\ge F_X(t)F_Y(t)=\Prob(M^{\ind}\le t),
			\qquad t\ge 0.
			\label{eq:diagcompare}
		\end{equation}
		Hence
		\begin{equation}
			M\le_{\st} M^{\ind},
			\label{eq:storder}
		\end{equation}
		where $\le_{\st}$ denotes the usual stochastic order. Consequently,
		\begin{equation}
			\E[M^r]\le \E[(M^{\ind})^r], \qquad r>0,
			\label{eq:momentcompare}
		\end{equation}
		and, whenever both queues are stable at arrival rate $\nu$,
		\begin{equation}
			\E[W^{\dep}]\le \E[W^{\ind}].
			\label{eq:waitingcompare}
		\end{equation}
	\end{proposition}
	
	\begin{proof}
		Inequality \eqref{eq:diagcompare} follows from the positive quadrant dependence property
		\eqref{eq:pqd}. This implies \eqref{eq:storder}. The moment inequalities in
		\eqref{eq:momentcompare} then follow from a standard property of the usual stochastic
		order; see \cite{ShakedShanthikumar2007}. 
		Finally, \eqref{eq:waitingcompare} follows from \eqref{eq:PK}, since $\E[W]$ is increasing in
		$\E[M]$ and $\E[M^2]$.
		\end{proof}
	
	\begin{remark}
		The comparison in Proposition~\ref{prop:benchmark} is a same-marginal benchmark comparison:
		the distributions of $X$ and $Y$ are held fixed, and only their dependence structure is changed.
		Accordingly, the inequalities in \eqref{eq:storder}--\eqref{eq:waitingcompare} isolate the effect of dependence alone. In particular, for the class considered here, replacing the dependent pair $(X,Y)$ by an independent pair with the same marginals yields a conservative approximation in the sense that it overestimates completion times and mean waiting times.
	\end{remark}
	
	%%%%%%%%%%%%%%%%%%%%%%%%%%%
	%%% SECTION 5 
	%%%%%%%%%%%%%%%%%%%%%%%%%
	\section{Numerical illustration}
	\label{sec:numerics}
	
	This section quantifies the waiting-time gap implied by Proposition~\ref{prop:benchmark} between the dependent model and the same-marginal independence benchmark across a range of parameter configurations. The dependent-model moments are computed from Lemma~\ref{lem:moments}. For the independence benchmark, the moments of $M^{\ind}$ are computed exactly using the matrix-exponential representation of ~\cite{FioriniLipsky2015}. All numerical results are therefore based on exact moment calculations.

	\subsection{Verification of Lemma~\ref{lem:moments}}
	
	Table~\ref{tab:verification} compares the closed-form expressions
	\eqref{eq:EMclosed}--\eqref{eq:EM2closed} with direct numerical quadrature of the
	tail-integral formulas in \eqref{eq:momentformulas}, using $G_M(t)$ from
	Theorem~\ref{thm:main}, for the baseline family
	\[
	\alpha_1=\alpha_2=2,\qquad \lambda_1=\lambda_2=2,\qquad a+b=c+d=1,\qquad a>c.
	\]
	The two approaches—our exact formulas and direct numerical quadrature—agree to machine precision throughout.
	
	\begin{table}[!h]
		\centering
		\caption{Verification of Lemma~\ref{lem:moments}. Closed-form (CF) values are from
			\eqref{eq:EMclosed}--\eqref{eq:EM2closed}; quadrature (Q) values integrate
			$G_M(t)$ from Theorem~\ref{thm:main}.}
		\label{tab:verification}
		\footnotesize
		\setlength{\tabcolsep}{5pt}
		\begin{tabular}{@{}ccccccc@{}}
			\toprule
			$(a,c)$ & $\Corr(X,Y)$ &
			$\E[M]_{\mathrm{CF}}$ & $\E[M]_{\mathrm{Q}}$ &
			$\E[M^2]_{\mathrm{CF}}$ & $\E[M^2]_{\mathrm{Q}}$ & Error \\
			\midrule
			$(0.55,0.45)$ & 0.980 & 1.03750 & 1.03750 & 1.34625 & 1.34625 & $<10^{-13}$ \\
			$(0.65,0.35)$ & 0.835 & 1.11250 & 1.11250 & 1.55375 & 1.55375 & $<10^{-13}$ \\
			$(0.75,0.25)$ & 0.600 & 1.18750 & 1.18750 & 1.78125 & 1.78125 & $<10^{-13}$ \\
			$(0.85,0.15)$ & 0.342 & 1.26250 & 1.26250 & 2.02875 & 2.02875 & $<10^{-13}$ \\
			$(0.95,0.05)$ & 0.105 & 1.33750 & 1.33750 & 2.29625 & 2.29625 & $<10^{-13}$ \\
			\bottomrule
		\end{tabular}
	\end{table}
	
	\subsection{Magnitude of the benchmark gap}
	
	Table~\ref{tab:baseline} reports $\E[W^{\dep}]$ and $\E[W^{\ind}]$ for the symmetric parameter family described above at arrival rate $\nu=0.45$. The relative gap,
	\[
	\frac{\E[W^{\ind}] - \E[W^{\dep}]}{\E[W^{\dep}]},
	\]
	ranges from $4.0\%$ at correlation $0.105$ to $73.8\%$ at correlation $0.980$. Hence, the independence benchmark is conservative across all specifications considered, and its bias grows sharply as dependence strengthens.

		Table~\ref{tab:loadrobust} extends the comparison to $\nu\in\{0.25,0.45,0.65\}$. The same qualitative pattern holds at each load level, with the conservative bias becoming larger as the system becomes more heavily loaded.  This sensitivity to load, however, is itself governed by the strength of dependence: for weak
	dependence ($\Corr=0.105$), the relative gap grows modestly from $2.4\%$ to $12.6\%$ as
	$\nu$ increases from $0.25$ to $0.65$, whereas for strong dependence ($\Corr=0.980$) it
	grows from $51.1\%$ to $163.9\%$ over the same range. Thus heavier traffic
	disproportionately amplifies the independence bias when dependence is strong.

	\begin{table}[!t]
		\centering
		\caption{Baseline symmetric family: $\alpha_1=\alpha_2=2$,
			$\lambda_1=\lambda_2=2$, $a+b=c+d=1$, $a>c$, $\nu=0.45$. Waiting times are computed
			from Lemma~\ref{lem:moments} and \cite{FioriniLipsky2015} via \eqref{eq:PK}.}
		\label{tab:baseline}
		\footnotesize
		\setlength{\tabcolsep}{6pt}
		\begin{tabular}{@{}cccccc@{}}
			\toprule
			$(a,c)$ & $\Corr(X,Y)$ &
			$\E[W]_{\dep}$ & $\E[W]_{\ind}$ & Rel.\ gap \\
			\midrule
			$(0.55,0.45)$ & 0.980 & 0.568 & 0.987 & $+73.8\%$ \\
			$(0.60,0.40)$ & 0.923 & 0.631 & 1.001 & $+58.6\%$ \\
			$(0.65,0.35)$ & 0.835 & 0.700 & 1.023 & $+46.1\%$ \\
			$(0.70,0.30)$ & 0.724 & 0.776 & 1.054 & $+35.7\%$ \\
			$(0.75,0.25)$ & 0.600 & 0.861 & 1.094 & $+27.1\%$ \\
			$(0.80,0.20)$ & 0.471 & 0.954 & 1.143 & $+19.8\%$ \\
			$(0.85,0.15)$ & 0.342 & 1.057 & 1.202 & $+13.7\%$ \\
			$(0.90,0.10)$ & 0.220 & 1.171 & 1.271 & $+8.5\%$ \\
			$(0.95,0.05)$ & 0.105 & 1.298 & 1.349 & $+4.0\%$ \\
			\bottomrule
		\end{tabular}
	\end{table}
	
	\begin{table}[!t]
		\centering
		\caption{Relative gap $(\E[W]_{\ind}-\E[W]_{\dep})/\E[W]_{\dep}$ (\%)
			across three arrival rates. Same parameter family as Table~\ref{tab:baseline}.}
		\label{tab:loadrobust}
		\footnotesize
		\setlength{\tabcolsep}{6pt}
		\begin{tabular}{@{}ccccc@{}}
			\toprule
			$(a,c)$ & $\Corr(X,Y)$ & $\nu=0.25$ & $\nu=0.45$ & $\nu=0.65$ \\
			\midrule
			$(0.55,0.45)$ & 0.980 & 51.1\% & 73.8\% & 163.9\% \\
			$(0.65,0.35)$ & 0.835 & 31.4\% & 46.1\% & 106.3\% \\
			$(0.75,0.25)$ & 0.600 & 18.0\% & 27.1\% & 66.8\% \\
			$(0.85,0.15)$ & 0.342 & 8.8\% & 13.7\% & 37.5\% \\
			$(0.95,0.05)$ & 0.105 & 2.4\% & 4.0\% & 12.6\% \\
			\bottomrule
		\end{tabular}
	\end{table}

	\subsection{Robustness to asymmetric specifications}
	
	Table~\ref{tab:asym} reports results for ten asymmetric parameter configurations that vary the linear coefficients, Erlang shape parameters, and Erlang rates. These cases cover a wide range of dependence levels, with correlations between $0.150$ and $0.958$, and satisfy the model assumptions and stability condition at the stated arrival rate $\nu$.
	
	The same qualitative conclusion holds throughout. In every case, the independence benchmark overestimates the mean waiting time relative to the dependent model. The magnitude of this conservative bias varies substantially across specifications: in Cases A--C, the relative error ranges from $3.3\%$ to $8.8\%$, whereas in Cases H--J it ranges from $26.6\%$ to $47.7\%$. Thus, the size of the gap is not driven by correlation alone; it also depends on the marginal service-time parameters and the traffic intensity.

	\begin{table}[!t]
		\centering
		\caption{Asymmetric configurations ordered by increasing correlation.
			Dependent moments are from Lemma~\ref{lem:moments}; independent moments are from
			\cite{FioriniLipsky2015}. All cases satisfy \eqref{eq:crossing} and
			$\Delta=ad-bc>0$.}
		\label{tab:asym}
		\scriptsize
		\setlength{\tabcolsep}{3.5pt}
		\begin{tabular}{@{}clcccccc@{}}
			\toprule
			Case &
			\shortstack[l]{$(\alpha_1,\lambda_1;\,\alpha_2,\lambda_2)$} &
			$(a,b,c,d)$ &
			$\Corr$ & $\nu$ &
			$\E[W]_{\dep}$ & $\E[W]_{\ind}$ &
			Rel.\ gap \\
			\midrule
			A & $(2,2.0;\,2,2.0)$ & $(1.10,0.05,0.10,0.95)$ & 0.150 & 0.20 & 0.386 & 0.398 & $3.3\%$ \\
			B & $(2,2.0;\,2,2.0)$ & $(1.00,0.15,0.20,0.90)$ & 0.359 & 0.20 & 0.356 & 0.388 & $8.8\%$ \\
			C & $(3,2.5;\,1,1.3)$ & $(1.00,0.15,0.25,0.85)$ & 0.411 & 0.18 & 0.329 & 0.357 & $8.6\%$ \\
			D & $(1,1.5;\,3,3.0)$ & $(0.80,0.30,0.20,0.90)$ & 0.536 & 0.22 & 0.249 & 0.281 & $13.1\%$ \\
			E & $(2,2.0;\,2,2.0)$ & $(0.90,0.25,0.35,0.95)$ & 0.584 & 0.22 & 0.441 & 0.519 & $17.6\%$ \\
			F & $(2,2.0;\,2,2.0)$ & $(1.00,0.10,0.60,0.90)$ & 0.635 & 0.18 & 0.391 & 0.468 & $19.7\%$ \\
			G & $(3,2.2;\,2,1.8)$ & $(0.95,0.20,0.40,0.90)$ & 0.741 & 0.20 & 0.530 & 0.670 & $26.3\%$ \\
			H & $(2,2.0;\,2,2.0)$ & $(1.20,0.05,0.70,0.95)$ & 0.847 & 0.18 & 0.497 & 0.629 & $26.6\%$ \\
			I & $(2,2.0;\,2,2.0)$ & $(1.30,0.05,0.90,0.95)$ & 0.910 & 0.16 & 0.641 & 0.912 & $42.3\%$ \\
			J & $(4,3.0;\,2,1.7)$ & $(1.15,0.08,0.80,0.92)$ & 0.958 & 0.18 & 0.582 & 0.860 & $47.7\%$ \\
			\bottomrule
		\end{tabular}
	\end{table}

%	\bibliographystyle{elsarticle-num}
%	\bibliography{ORL-references}
%	

\section*{Appendix. Proof of Lemma~\ref{lem:moments}}

Throughout this appendix, $\mu$ and $\kappa$ are as defined in
\eqref{eq:mukappa}, and $S_k$ is as defined in \eqref{eq:Sk}.
We also use the standard identities
\begin{equation}
	\int_0^\infty u^{n-1} e^{-\kappa u}\,du
	=
	\frac{(n-1)!}{\kappa^n},
	\qquad n\ge 1,\ \kappa>0,
	\label{eq:app-gamma}
\end{equation}
and
\begin{equation}
	\int_\eta^\infty s^n e^{-\lambda_2 s}\,ds
	=
	e^{-\lambda_2 \eta}
	\sum_{i=0}^n
	\frac{n!}{i!}\,
	\frac{\eta^i}{\lambda_2^{\,n-i+1}},
	\qquad n\ge 0,\ \eta\ge 0,
	\label{eq:app-tail}
\end{equation}
both of which follow by repeated integration by parts.

\paragraph{Proof of \eqref{eq:EMclosed}.}
Starting from Theorem~\ref{thm:main}, write
\[
\E[M]=\int_0^\infty \bigl(1-G_M(t)\bigr)\,dt
\]
and using \eqref{eq:GMfinal}, decompose \(1-G_M(t)\) as
\begin{equation}
	1-G_M(t)
	=
	\int_0^{u^*(t)}
	\overline F_V\!\left(\frac{t-cu}{d}\right)f_U(u)\,du
	+
	\int_{u^*(t)}^{t/a}
	\overline F_V\!\left(\frac{t-au}{b}\right)f_U(u)\,du
	+
	\overline F_U\!\left(\frac{t}{a}\right),
	\label{eq:app-tail-decomp}
\end{equation}
where \(\overline F_V:=1-F_V\) and \(\overline F_U:=1-F_U\). Integrating over
\(t\in[0,\infty)\) gives
\[
\E[M]=J_1+J_2+J_3,
\]
where each \(J_i\) corresponds to one term of \eqref{eq:app-tail-decomp}.

\emph{Term \(J_3\).}
Substituting \(s=t/a\),
\[
J_3
=
\int_0^\infty \overline F_U\!\left(\frac{t}{a}\right)\,dt
=
a\int_0^\infty \overline F_U(s)\,ds
=
a\,\E[U]
=
\frac{a\alpha_1}{\lambda_1}.
\]

\emph{Term \(J_1\).}
The integration region is
\[
\{(u,t):0\le u\le u^*(t),\ t\ge 0\}.
\]
Since \(u^*(t)=(d-b)t/\Delta\), the condition \(u\le u^*(t)\) is equivalent to
\[
t\ge \frac{\Delta}{d-b}\,u.
\]
Applying Tonelli's theorem and then substituting \(s=(t-cu)/d\), so that
\(dt=d\,ds\), we obtain
\[
J_1
=
\int_0^\infty f_U(u)\int_{\Delta u/(d-b)}^\infty
\overline F_V\!\left(\frac{t-cu}{d}\right)\,dt\,du
=
d\int_0^\infty f_U(u)\int_{\mu u}^\infty \overline F_V(s)\,ds\,du,
\]
because
\[
\frac{\Delta u/(d-b)-cu}{d}
=
\frac{(ad-bc)u-c(d-b)u}{d(d-b)}
=
\frac{a-c}{d-b}\,u
=
\mu u.
\]

\emph{Term \(J_2\).}
The integration region is
\[
\{(u,t):u^*(t)\le u\le t/a,\ t\ge 0\},
\]
which is equivalent to
\[
\{(u,t):u\ge 0,\ au\le t\le \Delta u/(d-b)\}.
\]
Applying Tonelli's theorem and substituting \(s=(t-au)/b\), so that \(dt=b\,ds\),
gives
\[
J_2
=
\int_0^\infty f_U(u)\int_{au}^{\Delta u/(d-b)}
\overline F_V\!\left(\frac{t-au}{b}\right)\,dt\,du
=
b\int_0^\infty f_U(u)\int_0^{\mu u}\overline F_V(s)\,ds\,du,
\]
because
\[
\frac{\Delta u/(d-b)-au}{b}
=
\frac{(ad-bc)u-a(d-b)u}{b(d-b)}
=
\frac{a-c}{d-b}\,u
=
\mu u.
\]

\textcolor{blue}{Combining $J_1$ and $J_2$} and using
%Combining the three terms and using
\[
\int_0^\infty \overline F_V(s)\,ds=\E[V]=\frac{\alpha_2}{\lambda_2},
\]
we have
\[
J_1+J_2
=
d\int_0^\infty f_U(u)\int_{\mu u}^\infty \overline F_V(s)\,ds\,du
+
b\int_0^\infty f_U(u)\int_0^{\mu u}\overline F_V(s)\,ds\,du.
\]
Since
\[
\int_0^{\mu u}\overline F_V(s)\,ds
=
\E[V]-\int_{\mu u}^\infty \overline F_V(s)\,ds,
\]
it follows that
\[
J_1+J_2
=
b\,\E[V]
+
(d-b)\int_0^\infty f_U(u)\int_{\mu u}^\infty \overline F_V(s)\,ds\,du.
\]
Hence
\begin{equation}
	\E[M]
	=
	\frac{a\alpha_1}{\lambda_1}
	+
	\frac{b\alpha_2}{\lambda_2}
	+
	(d-b)\int_0^\infty f_U(u)\int_{\mu u}^\infty \overline F_V(s)\,ds\,du.
	\label{eq:app-EM-start}
\end{equation}

We now evaluate the double integral. Using the Erlang survival function,
\[
\overline F_V(s)
=
e^{-\lambda_2 s}\sum_{j=0}^{\alpha_2-1}\frac{(\lambda_2 s)^j}{j!},
\]
we obtain
\[
\int_{\mu u}^\infty \overline F_V(s)\,ds
=
\sum_{j=0}^{\alpha_2-1}\frac{\lambda_2^j}{j!}
\int_{\mu u}^\infty s^j e^{-\lambda_2 s}\,ds.
\]
Applying \eqref{eq:app-tail} with \(n=j\) and \(\eta=\mu u\),
\[
\int_{\mu u}^\infty s^j e^{-\lambda_2 s}\,ds
=
e^{-\lambda_2\mu u}
\sum_{i=0}^j
\frac{j!}{i!}\,
\frac{(\mu u)^i}{\lambda_2^{\,j-i+1}}.
\]
Therefore
\[
\int_{\mu u}^\infty \overline F_V(s)\,ds
=
e^{-\lambda_2\mu u}
\sum_{j=0}^{\alpha_2-1}\sum_{i=0}^j
\frac{(\mu u)^i}{i!}\,\lambda_2^{\,i-1}.
\]
Collecting terms by the power \(k=i\) of \(u\) yields
\[
\int_{\mu u}^\infty \overline F_V(s)\,ds
=
e^{-\lambda_2\mu u}
\sum_{k=0}^{\alpha_2-1}
\frac{(\alpha_2-k)\mu^k\lambda_2^{\,k-1}}{k!}\,u^k,
\]
since \(\sum_{j=k}^{\alpha_2-1}1=\alpha_2-k\).

Substituting
\[
f_U(u)=\frac{\lambda_1^{\alpha_1}}{(\alpha_1-1)!}u^{\alpha_1-1}e^{-\lambda_1u}
\]
into \eqref{eq:app-EM-start}, we obtain
\[
\int_0^\infty
f_U(u)e^{-\lambda_2\mu u}u^k\,du
=
\frac{\lambda_1^{\alpha_1}}{(\alpha_1-1)!}
\int_0^\infty
u^{\alpha_1+k-1}e^{-\kappa u}\,du
=
\frac{\lambda_1^{\alpha_1}}{(\alpha_1-1)!}
\cdot
\frac{(\alpha_1+k-1)!}{\kappa^{\alpha_1+k}}
\]
by \eqref{eq:app-gamma}. Hence
\[
\int_0^\infty
f_U(u)\int_{\mu u}^\infty \overline F_V(s)\,ds\,du
=
\frac{\lambda_1^{\alpha_1}}{(\alpha_1-1)!}
\sum_{k=0}^{\alpha_2-1}
\frac{(\alpha_2-k)\mu^k\lambda_2^{\,k-1}}{k!}
\cdot
\frac{(\alpha_1+k-1)!}{\kappa^{\alpha_1+k}}.
\]
Substituting this into \eqref{eq:app-EM-start} gives \eqref{eq:EMclosed}.

\paragraph{Proof of \eqref{eq:EM2closed}.}
Starting again from Theorem~\ref{thm:main}, write
\[
\E[M^2]
=
2\int_0^\infty t\bigl(1-G_M(t)\bigr)\,dt
\]
and decompose \(1-G_M(t)\) as in \eqref{eq:app-tail-decomp}. This gives
\[
\E[M^2]=2(K_1+K_2+K_3),
\]
where
\[
K_i=\int_0^\infty t\cdot (\text{\(i\)-th term of \eqref{eq:app-tail-decomp}})\,dt.
\]

\emph{Term \(K_3\).}
Substituting \(s=t/a\),
\[
2K_3
=
2\int_0^\infty t\,\overline F_U\!\left(\frac{t}{a}\right)\,dt
=
2a^2\int_0^\infty s\,\overline F_U(s)\,ds
=
a^2\E[U^2]
=
\frac{a^2\alpha_1(\alpha_1+1)}{\lambda_1^2}.
\]

\emph{Term \(K_1\).}
Applying Tonelli's theorem over the region
\[
\{(u,t):u\ge 0,\ t\ge \Delta u/(d-b)\}
\]
and substituting \(s=(t-cu)/d\), so that \(t=ds+cu\) and \(dt=d\,ds\), yields
\[
K_1
=
\int_0^\infty f_U(u)\int_{\Delta u/(d-b)}^\infty
t\,\overline F_V\!\left(\frac{t-cu}{d}\right)\,dt\,du
\]
\[
=
\int_0^\infty f_U(u)\int_{\mu u}^\infty (ds+cu)\,\overline F_V(s)\,d\,ds\,du
\]
\[
=
d^2\int_0^\infty f_U(u)\int_{\mu u}^\infty s\,\overline F_V(s)\,ds\,du
+
dc\int_0^\infty u f_U(u)\int_{\mu u}^\infty \overline F_V(s)\,ds\,du.
\]

\emph{Term \(K_2\).}
Applying Tonelli's theorem over the region
\[
\{(u,t):u\ge 0,\ au\le t\le \Delta u/(d-b)\}
\]
and substituting \(s=(t-au)/b\), so that \(t=bs+au\) and \(dt=b\,ds\), gives
\[
K_2
=
\int_0^\infty f_U(u)\int_{au}^{\Delta u/(d-b)}
t\,\overline F_V\!\left(\frac{t-au}{b}\right)\,dt\,du
\]
\[
=
b^2\int_0^\infty f_U(u)\int_0^{\mu u}s\,\overline F_V(s)\,ds\,du
+
ab\int_0^\infty u f_U(u)\int_0^{\mu u}\overline F_V(s)\,ds\,du.
\]

Define
\[
\theta_s
:=
\int_0^\infty f_U(u)\int_{\mu u}^\infty s\,\overline F_V(s)\,ds\,du,
\qquad
\theta_u
:=
\int_0^\infty u f_U(u)\int_{\mu u}^\infty \overline F_V(s)\,ds\,du.
\]
Using
\[
2b^2\int_0^\infty f_U(u)\int_0^{\mu u}s\,\overline F_V(s)\,ds\,du
=
b^2\E[V^2]-2b^2\theta_s
\]
and
\[
2ab\int_0^\infty u f_U(u)\int_0^{\mu u}\overline F_V(s)\,ds\,du
=
ab\,\E[UV]-2ab\,\theta_u,
\]
together with \(\E[V^2]=\alpha_2(\alpha_2+1)/\lambda_2^2\) and
\(\E[UV]=\E[U]\E[V]=\alpha_1\alpha_2/(\lambda_1\lambda_2)\), we obtain
\begin{equation}
	\E[M^2]
	=
	\frac{a^2\alpha_1(\alpha_1+1)}{\lambda_1^2}
	+
	\frac{b^2\alpha_2(\alpha_2+1)}{\lambda_2^2}
	+
	\frac{2ab\,\alpha_1\alpha_2}{\lambda_1\lambda_2}
	+
	2(d^2-b^2)\theta_s
	+
	2(dc-ab)\theta_u.
	\label{eq:app-EM2-start}
\end{equation}

It remains to evaluate \(\theta_s\) and \(\theta_u\).

For \(\theta_s\), using the Erlang survival function,
\[
\int_{\mu u}^\infty s\,\overline F_V(s)\,ds
=
\sum_{j=0}^{\alpha_2-1}\frac{\lambda_2^j}{j!}
\int_{\mu u}^\infty s^{j+1}e^{-\lambda_2 s}\,ds.
\]
Applying \eqref{eq:app-tail} with \(n=j+1\) and \(\eta=\mu u\),
\[
\int_{\mu u}^\infty s^{j+1}e^{-\lambda_2 s}\,ds
=
e^{-\lambda_2\mu u}
\sum_{i=0}^{j+1}
\frac{(j+1)!}{i!}\,
\frac{(\mu u)^i}{\lambda_2^{\,j+2-i}}.
\]
Hence
\[
\int_{\mu u}^\infty s\,\overline F_V(s)\,ds
=
e^{-\lambda_2\mu u}
\sum_{j=0}^{\alpha_2-1}\sum_{i=0}^{j+1}
(j+1)\frac{(\mu u)^i}{i!}\lambda_2^{\,i-2}.
\]
Collecting terms by \(k=i\) gives
\[
\int_{\mu u}^\infty s\,\overline F_V(s)\,ds
=
e^{-\lambda_2\mu u}
\sum_{k=0}^{\alpha_2}
\frac{S_k\mu^k\lambda_2^{\,k-2}}{k!}\,u^k,
\]
where
\[
S_k=\sum_{j=\max(0,k-1)}^{\alpha_2-1}(j+1).
\]
Multiplying by \(f_U(u)\) and integrating via \eqref{eq:app-gamma} yields
\[
\theta_s
=
\frac{\lambda_1^{\alpha_1}}{(\alpha_1-1)!}
\sum_{k=0}^{\alpha_2}
\frac{S_k\mu^k\lambda_2^{\,k-2}}{k!}
\cdot
\frac{(\alpha_1+k-1)!}{\kappa^{\alpha_1+k}}.
\]

For \(\theta_u\), from the derivation of \eqref{eq:EMclosed},
\[
\int_{\mu u}^\infty \overline F_V(s)\,ds
=
e^{-\lambda_2\mu u}
\sum_{k=0}^{\alpha_2-1}
\frac{(\alpha_2-k)\mu^k\lambda_2^{\,k-1}}{k!}\,u^k.
\]
Multiplying by \(u f_U(u)\) and integrating via \eqref{eq:app-gamma} gives
\[
\theta_u
=
\frac{\lambda_1^{\alpha_1}}{(\alpha_1-1)!}
\sum_{k=0}^{\alpha_2-1}
\frac{(\alpha_2-k)\mu^k\lambda_2^{\,k-1}}{k!}
\cdot
\frac{(\alpha_1+k)!}{\kappa^{\alpha_1+k+1}}.
\]

Substituting the expressions for \(\theta_s\) and \(\theta_u\) into
\eqref{eq:app-EM2-start} gives \eqref{eq:EM2closed}. \qed

\end{document}